\documentclass[12pt]{article}

\usepackage{amsmath,amsfonts,amssymb,amsthm}
\usepackage{color, graphicx}
\usepackage{epsfig,fullpage}
\usepackage{natbib}
\usepackage{url}
\usepackage{lineno}
\usepackage{lscape}
\usepackage{multirow}

\newcommand{\ul}{\underline}
\newcommand{\BEAS}{\begin{eqnarray*}}
\newcommand{\EEAS}{\end{eqnarray*}}
\newcommand{\BEA}{\begin{eqnarray}}
\newcommand{\EEA}{\end{eqnarray}}
\newcommand{\BIT}{\begin{itemize}}
\newcommand{\EIT}{\end{itemize}}
\newcommand{\BNUM}{\begin{enumerate}}
\newcommand{\ENUM}{\end{enumerate}}

\newcommand{\mZ}{\mathbf{Z}}

\newtheorem{theorem}{Theorem}

\begin{document}

\title{Regression with Partially Observed Ranks on a Covariate: \\
Distribution-Guided Scores for Ranks}

\date{}

\date{}
\author{Yuneung Kim, Johan Lim, Young-Geun Choi\\
and \\
Sujung Choi, and Do Hwan Park\thanks{
Yuneung Kim, Johan Lim and Young-Geun Choi 
are with the Department of Statistics, Seoul National University, Seoul, 151-747, Korea. 
 Sujung Choi is with the School of Business Administration, Soongsil University, Seoul, 156-743, Korea.
 Dohwan Park is with Department of Mathematics and Statistics, University of Maryland at Baltimore 
 County, Baltimore, MD, 21250, USA. 
 All correspondence  are to Johan Lim (E-mail: \texttt{johanlim@snu.ac.kr})}}
\maketitle
\begin{abstract}
\noindent This work is motivated by a hand-collected data set from one of the largest
Internet portals in Korea. This data set records the top 30 most frequently
discussed stocks on its on-line message board. The frequencies are 
considered to measure the attention paid by investors to individual stocks. The empirical goal
of the data analysis is to investigate the effect of this attention on trading behavior.
For this purpose, we regress the (next day) returns and the (partially) observed ranks of
frequencies. In the regression, the ranks are transformed into scores, for which purpose the
 identity or linear scores are commonly used.
In this paper, we propose a new class of scores 
(a score function) that is based on the moments of order statistics of a pre-decided random 
variable. The new score function, denoted by D-rank, is shown to be asymptotically optimal to 
maximize the correlation between 
the response and score, when the pre-decided random variable and true covariate
are in the same location-scale family. In addition, the least-squares estimator using the D-rank 
consistently estimates the true correlation between
the response and the covariate, and asymptotically approaches the normal distribution.
We additionally propose a procedure for diagnosing a given score function (equivalently,
the pre-decided random variable Z) and selecting
one that is better suited to the data. 
We numerically demonstrate the advantage of using a correctly specified score function 
over that of the identity scores (or other misspecified scores) 
in estimating the correlation coefficient. Finally, we apply our proposal
to test the effects of investors' attention on their returns using the motivating 
data set. 

\end{abstract}

\baselineskip 24pt

\section{Introduction}

This paper is motivated by a hand-collected data set from \texttt{Daum.net}, the 
2nd largest Internet portal in Korea. The \texttt{Daum.net} portal
 offers an on-line stock message board where investors can
 freely discuss specific stocks in which they might be
 interested. This portal also reports
a ranked list of the top 30 stocks that are most frequently discussed by users on a daily basis. 
The data set was collected by the authors during the 537 trading days 
from October 4th, 2010, to November 23rd, 2012. Along with the rank data, we also 
collected financial data regarding individual companies from FnGuide 
(\texttt{http://www.fnguide.com}).  These additional data include stock-day trading volumes 
classified in terms of different types of investors, stock prices, stock returns, and so on.

The purpose of analyzing the collected data is to investigate the shifts in stock returns caused by variations in
 investor attention. In finance, researchers are often interested in
 determining the motivations that drive buying and selling decisions in stock markets. It is commonly assumed that investors efficiently process
 relevant information in a timely manner, but in reality, it is nearly impossible to be efficient
because of information overload. In particular, individual investors are often less sophisticated 
than are institutional investors and have a limited ability to
process all relevant information. For this reason, individual investors may pay attention only to a limited amount 
of information, perhaps that which is relatively easy to access. The phenomenon of limited attention is a well-documented cognitive bias
 in the psychological literature
 \citep{Kahneman:1973,Camerer:2003}. This phenomenon affects the information-processing capacities of investors and thus may affect asset prices on the financial market. 
To empirically prove the 
effect of investor attention on stock returns, we regress the observed stock returns with respect to the partially observed ranks.

Regression on a (partially observed) rank covariate has not 
previously been extensively studied in the literature. A procedure 
that is commonly used in practice to address rank covariates is to (i) regroup the ranks into only a few groups 
(if the number of ranks is high) and (ii) treat the regrouped ranks as an ordinal categorical variable.
Ordered categorical variables frequently arise in various applications and have been 
studied extensively in the literature. Score-based analysis is most commonly used for this purpose;
see \cite{Hajek:1968}, \citet{Hora:1984}, \citet{Kimeldorf:1992},
\citet{Zheng:2008}, \citet{Gertheiss:2014} and the references therein. Thus, this typical two-step procedure for addressing a rank covariate is equivalent to defining a score function for the ranks. However, as in the case of ordinal categorical 
variables, such a score-based approach suffers from an inherent drawback related to the choice of the score function; 
different choices of scores may lead to conflicting conclusions in the analysis 
\citep{Graubard:1987,Ivanova:2001,Senn:2007}. The recommendation 
for selecting the score function according to the literature is (i) to choose meaningful scores for 
the ordinal categorical variable based on domain knowledge of the data, (ii) to 
use equally spaced scores if scientifically plausible scores are not available (see \citet{Graubard:1987}),
 and (iii) to find a optimal scaling transformed scores that maximize the correlation with 
 the responses while preserving the assumed characteristics of the ordinal values\citep{Linting:2007,Costantini:2010,
 deLeeuw:2009,Mair:2010,Jacoby:2016}.

In this paper, we seek to provide an efficient tool for approach (i) described above, for the case in which some qualitative 
knowledge is available regarding the ranks or the ranking variable (the variable that is ranked). More specifically,
we propose a new set of score functions, denoted by D-rank, and study their use in linear regression.  
The proposed score function is based on the moments of order statistics (MOS) of a pre-decided random variable
$Z$. This score function has several interesting properties related with the regression model, if 
the pre-decided random variable is correctly specified as listed below. Here, the correct specification implies 
it is within the same location-scale family with the true (unobserved) covariate $X$. 
First, the D-rank is asymptotically
optimal in the sense that it maximizes the correlation between the response and score if 
the distribution of the D-rank is correctly specified. Second, 
the least-squares estimator using the D-rank consistently estimates the true correlation between
the response and the covariate and asymptotically approaches the normal distribution. Finally,
the residuals of the fitted regression allow us to diagnose the given score function (equivalently,
the pre-decided random variable $X$) and to provide a tool for selecting a score function 
that is better suited to the data.

The remainder of this paper is organized as follows. In Section 2, we study the 
properties of the proposed D-rank. In this section, we show that the proposed D-rank
is asymptotically optimal to maximize the correlation between the response and score. 
In addition, We also demonstrate the 
asymptotic equivalence between the proposed score function and the quantile 
function; the quantile function may provide a better illustration of the 
qualitative features of the score function. In Section 3, we apply 
the score function to estimate the regression coefficient of the linear model or, 
more precisely, to estimate the correlation coefficient between 
the response and the scoring variable $X$.  
We prove that the least-squares estimator using the D-rank 
consistently estimates the correlation coefficient and is asymptotically normally 
distributed. In addition, we discuss the procedure for selecting an appropriate score function using 
the residuals.  
In Section 4, we numerically demonstrate that using the correctly specified score function 
significantly reduces the mean square error on the estimation of the correlation coefficient. 
In Section 5, we analyze the motivating data set to investigate the existence of the attention 
effect. Finally, in Section 6, we briefly summarize the paper and discuss 
the application of the proposed scores to regression using other auxiliary covariates.

\section{Distribution-Guided Scores for Ranks (D-rank)} 

We consider a simple regression model in which only partial ranks of a covariate are observed. Specifically, suppose that $\big\{ \big( Y_{i}, X_{i} \big), ~i=1,2,\ldots,n \big\}$ is the complete set of observations,
where $Y_{i}$ is the variable of primary interest and $X_{i}$ is the covariate related to $Y_i$. For example,
in our rank data from \texttt{Daum.net}, for $i=1,2,\ldots,n$, $Y_i$
is a relevant outcome such as earning rate or trading volume,
$X_i$ is the ``unobserved'' investors' attention on the $i$th company measured by the frequency of on-line discussions, and $R_i$ is the ``observed'' rank of $X_i$ among $X_1,X_2,\ldots,X_n$.  We 
make certain assumptions regarding the distributions of $X$ and $Y$. 
We assume that the linear model of the relationship between $X_i$ and $Y_i$ is 
\begin{equation}  \label{eqn:lm}
Y_i = \mu_Y + \rho \sigma_Y \frac{X_i - \mu_X}{\sigma_X} + \epsilon_i,
\end{equation}
where the $\epsilon_i$s are IID values from a distribution of mean $0$ and variance $\sigma_{\epsilon}^2$.
The objective of this paper is the estimation and inference of $\rho={\rm corr} \big(Y, X\big)$ (or
the regression coefficient between $Y$ and $X$) based on the observed data $\big\{ (Y_i, R_i), i=1,2,\ldots,n \big\}$. 
To do it, we aim to define a good score function $S(r)$ for the observed rank $r$, and 
consider the regression of $Y_{[r:n]}$ on $S(r)$, where $Y_{[r:n]}$ is the response $Y_i$ for $R_i=r$.

The D-rank, we propose in this paper, is a set of the MOS of pre-decided random variable $Z$, 
which we assume is in the same location-scale family of the true covariate $X$. To be specific, 
suppose that $Z_1,Z_2,\ldots,Z_n$ are independent and identically distributed (IID) 
copies of the random variable $Z$ and that $Z_{(r:n)}$ is the corresponding $r$th-order statistic for $r=1,2,\ldots,n$.
The D-rank defines the score of the rank $r$ as $S_n(r)=\alpha_{(r:n)}:={\rm E} \big(Z_{(r:n)} \big)$ for $r=1,2,\ldots,n$.

We first show that the D-rank maximizes the sample correlation between $Y_{[r:n]}$ and $\alpha_{(r:n)}$, $r=1,2,\ldots,n$, 
in asymptotic, among all increasing functions $S_n(r):\{1,2,\ldots,n\} \rightarrow \mathbb{R}$. 
Let  $\overline{S}_n(r)$ and $\overline{\alpha}_{(r:n)}$ be the standardized scores (of $S_n(r)$ 
and $\alpha_{(r:n)}$) to make $\sum_{r=1}^n \overline{S}_n(r) = \sum_{r=1}^n
\overline{\alpha}_{(r:n)}=0$ and $\sum_{r=1}^n \overline{S}^2_n(r) = \sum_{r=1}^n\overline{\alpha}^2_{(r:n)}=1$.
Let ${\bf S}_n$ and $\overline{\bf S}_n$ be the collection of all increasing functions $S_n(r)$ and $\overline{S}_n(r)$, 
respectively. 
 \begin{theorem}
Under the linear model (\ref{eqn:lm}), if $Z$ is in the location-scale family of $X$, the D-rank maximizes 
the limit of the sample correlation between $Y_{[r:n]}$ and $S_n(r)$ among $S_n(r) \in {\bf S}_n$: 
\begin{equation} \label{eqn:acorr} 
\lim_{n \rightarrow \infty} \frac{1}{n \cdot \widehat{\sigma}_Y} \sum_{r=1}^n \overline{S}_n(r) \big( Y_{[r:n]} - \overline{Y}_n \big),
\end{equation} 
 where $\widehat{\sigma}^2_Y = \frac{1}{n}\sum_{r=1}^n  \big( Y_{[r:n]} - \overline{Y}_n \big)^2$.
 \end{theorem}
\noindent The proof of Theorem 1 is followed in Appendix.

Theorem 1 shows the asymptotic optimality of the D-rank for the regression in
view of optimal scaling in the literature. The optimal scaling finds optimally transformed scores
that explain mostly well the assumed statistical model. It arises in various contexts
including Gifi classification of non-linear multivariate analysis\citep{deLeeuw:2009},
the aspect (correlational and non-cprrelational aspects) of multivariable\citep{Mair:2010}, 
and non-linear principal component analysis\citep{Linting:2007,Costantini:2010}. 
Here, we adopt the idea of the optimal scaling in \citet{Jacoby:2016}, and find the transformation to 
maximize the correlation 
between the response and transformed scores. Theorem 1 above shows that the D-rank maximizes 
the correlation in asymptotic, if pre-determined distribution for the D-rank is correctly specified.

The proposed score is closely related to the quantile of the underlying distribution of $Z$. Let 
$F_{Z} (z)$ for $z \in \mathbb{R}$ and $Q_Z(q)$ for $q \in [0,1]$ be 
the cumulative distribution function (CDF) and the quantile function (QF), respectively, of $Z$. 
In the estimation of $F_Z(z)$ for $\{Z_i, i=1,2,\ldots,n\}$, 
the $r$th-order statistic $Z_{(r:n)}$ is the $(r/n)\times 100$-th percentile point of the empirical CDF, 
and thus, its expected value is approximately equal to $Q_Z(r/n)$. More specifically, given $p_r=\frac{r}
{n+1}$, $q_r=1-p_r$, and $Q_r=Q_Z(p_r)$, we can write 
\begin{equation} \nonumber 
\alpha_{(r:n)}=Q_r+ \frac{p_r q_r}{2(n+2)} Q_r^{(2)}+O\big(\frac{1}{n^2} \big),
\end{equation}
where $Q_r^{(2)}=-{f_Z^{\prime} (Q_r)} \big/{\{ f_Z (Q_r) \}^3}$ and $f_Z(z)$ is the probability density 
function of $Z$, which is differentiable. We refer the reader to \citet[Section 4.6]{David:2003} 
for the details of the relationship between the MOS and the quantiles.

Consideration of the QF may provide a better understanding of the qualitative features of the proposed score 
function. Suppose we expect the score function $S_n(r)$ is convex in tail 
(for $r \ge [n c]$ for a constant $c$ close to $1$); in other words, 
$S_n (r+1) - S_n (r) \ge S_n(r) - S_n (r-1)$ for $r \ge [n c]$. 
From the equivalence between the MOS and quantiles,
it is known that the convexity of the scores $S_n(r)$ is approximately equal to that 
of the quantile function $Q_Z(p)$. Furthermore, 
the convexity of $Q_Z(p)$ for $p \ge c$  implies the following 
equivalent statements: (i) 
$F(z)$ is concave in $z$, (ii) $f^{\prime} (z) \le 0$ or (iii) $\log f (z)$ is decreasing in $z$, all for $z \ge Q_Z([nc])$.

\section{Simple Linear Regression}

In this section, we consider a simple regression model in which only partial ranks of a covariate are observed. Specifically, suppose that $\big\{ \big( Y_{i}, X_{i} \big), ~i=1,2,\ldots,n \big\}$ is the complete set of observations
from the linear model (\ref{eqn:lm}),
and $R_i$ is the rank of $X_i$ among $X_1,X_2,\ldots,X_n$. The rank $R_i$ 
of $X_i$ is indirectly measured by the frequency of on-line discussions of the $i$th company.

In this paper,
 we consider the case in which the ranks $R_i$ are partially observed 
in the sense that we observe only that
$U_{i}=R_i {\rm I}\big(R_{i} \le m \big) + m^{+} {\rm I} \big( R_{i}>m \big)$
rather than $R_i$, where $m^+$ is an arbitrary constant that is greater than $m$. Finally, 
the observations are 
\begin{equation} \nonumber
\big\{ (Y_i, U_i), ~i=1,2,\ldots,n \big\}.
\end{equation}
We let $Y_{[r:n]}=Y_{i} {\rm I} \big(R_{i}=r\big)$ 
for $r=1,2,\ldots,m$, and denote the above partially observed data by 
${\bf Y}_{[m]}$ for notational simplicity.

The objective of this section is to identify a good estimator of $\rho={\rm corr} \big(Y, X\big)$ (or
the regression coefficient between $Y$ and $X$) and to test $\mathcal{H}_0: \rho=0$ versus $\mathcal{H}_1: \rho\neq 0$ or $\rho >0$ using the observed data ${\bf Y}_{[m]}$.

\subsection{Least-Squares Estimator}

To estimate $\rho$, we recall assumptions regarding the distributions of $X$ and $Y$. 
We assume that the linear model of the relationship between $X_i$ and $Y_i$ is 
\begin{equation}  \nonumber
Y_i = \mu_Y + \rho \sigma_Y \frac{X_i - \mu_X}{\sigma_X} + \epsilon_i,
\end{equation}
where the $\epsilon_i$s are IID values from a distribution of mean $0$ and variance $\sigma_{\epsilon}^2$. 
By ordering on the $X_i$s, we have for $r=1,\ldots,n$
\begin{equation}
Y_{[r:n]} = \mu_Y + \rho \frac{\sigma_Y}{\sigma_X} \big( X_{(r:n)} -\mu_X \big) + \epsilon_{[r:n]}, \label{eqn:con-linear}
\end{equation}
where $\rho = {\rm corr} \big( Y, X \big)$ and 
\begin{eqnarray}
{\rm E} \big(Y_{[r:n]} \big) &=& \mu_Y + \rho \sigma_Y \alpha_{(r:n)}  \label{eqn:con-mean}\\
{\rm var} \big( Y_{[r:n]} \big) &=& \sigma_Y^2 \big( \rho^2 \beta_{(rr:n)} + 1 - \rho^2 \big) \nonumber \\
{\rm cov} \big(Y_{[r:n]},Y_{[s:n]}\big) &=& \rho^2\sigma_Y^2\beta_{(rs:n)},~r \ne s \nonumber
\end{eqnarray}
with 
\begin{equation}  \nonumber
\alpha_{(r:n)} = {\rm E} \bigg\{ \frac{ X_{(r:n)} - \mu_X } { \sigma_X} \bigg\} \quad \mbox{and } ~~
\beta_{(rs:n)} = {\rm Cov} \left(   \frac{ X_{(r:n)} - \mu_X } { \sigma_X},
\frac{ X_{(s:n)} - \mu_X } { \sigma_X} \right)
\end{equation}
for $r,s=1,2,\ldots,n$ \citep{David:1974,David:2003}.

We are motivated by the identities (\ref{eqn:con-linear}) and (\ref{eqn:con-mean}) given above and
propose the least-squares estimator 
\begin{equation} \label{eqn:lse}
\widehat{\rho} \big(s \big) \equiv \frac{1}{\widehat{\sigma}_Y} \cdot
 \frac{\sum_{r=1}^{[ns]} \alpha_{(r:n)}\big\{  Y_{[r:n]} -\widehat{\mu}_Y  \big\}}
{\sum_{r=1}^{[ns]} \alpha_{(r:n)}^2 }
\end{equation}
as an estimator of $\rho$ with $s=m/n$, where, $\widehat{\mu}_Y=\sum_{i=1}^nY_i/n$ and $\widehat{\sigma}_Y^2=\sum_{i=1}^n(Y_i-\widehat{\mu}_Y)^2/n$
are the empirical estimators of the mean and variance, respectively, of $Y$.

We claim that, if 
$X$ is drawn from a location-scale family generated by $Z$, then the least-squares estimator 
$ \widehat{\rho} \big(s \big)$ with $s=m/n$ in (\ref{eqn:lse}), that is calculated
based on the partial observations ${\bf Y}_{[m]}$,
is consistent and asymptotically normally distributed with an appropriate scale, 
as shown in Theorem 2. Suppose that
\begin{equation} \nonumber
\Psi_n^{\rm I}(s) := \frac{1}{n} \sum_{r=1}^{[ns]} \alpha_{(r:n)}^2 \sigma_{(r:n)}^2, 
\Psi_n^{\rm II}(s):= \frac{1}{n} \sum_{r1=1}^{[ns]} \sum_{r2=1}^{[ns]} \alpha_{(r1:n)}
 \alpha_{(r2:n)}\beta_{(r1,r2:n)}^2,~ \mbox{and}~ \Phi_n(s):=  \frac{1}{n} \sum_{r=1}^{[ns]}
 \alpha_{(r:n)}^2,
\end{equation}
where $\sigma_{(r:n)}^2=
\sigma^2\big(X_{(r:n)} \big)$, and let $\Psi_{\infty}^{\rm I}(s)$, $\Psi_{\infty}^{\rm II}(s)$ and $\Phi_{\infty}(s)$ be the limits of 
 $\Psi_n^{\rm I}(s)$, $\Psi_n^{\rm II}(s)$ and $\Phi_n(s)$, respectively (under the assumption that they exist). 

\begin{theorem}
Under the assumption that $X$ is drawn from a distribution of a location-scale family with a finite variance,
the distribution of $\sqrt{n} \big( \widehat{\rho} (s) -\rho \big) $ converges to the normal distribution of
 mean $0$ and variance $\big\{\Psi_{\infty}^{\rm I}(s)/\sigma_Y^2+\rho^2\Psi_{\infty}^{\rm II}(s)\big\}/\Phi_{\infty}^2(s)$.
\end{theorem}
\noindent The proof of Theorem 2 is provided in the Appendix.

We conclude this section with two remarks regarding Theorem 2. First, in Theorem 2, 
from the tower property of the conditional expectation, 
\begin{equation}\nonumber
{\rm var}\big(\sqrt{n} \widehat{\rho} \big) > \frac{1}{ \big(1 \big/ n \big) \sum_{r=1}^{[ns]}
 \alpha_{(r:n)}^2} \ge \frac{1}{{\rm var} \big(X \big)}=1,
\end{equation}
and when $\rho=0$, the asymptotic variance of $\sqrt{n} \widehat{\rho}$ is larger than $1$,
which is the variance of the least-squares estimator in the case where $X$ is completely observed. Second, 
it is possible to test the hypothesis $\mathcal{H}_0:\rho=0$ using the statistic 
$
{\rm T}= \sqrt{n}  \widehat{\rho},
$
which has an asymptotically normal distribution of mean $0$ and variance $1\big/\Phi_{\infty}(s)$.

\subsection{Residual Analysis}

As in the classical linear model, the residuals can provide guidance 
for identifying a better model and score 
function. The residuals are defined as $e_{[r:n]}^* 
= \big(Y_{[r:n]}-\mu_Y\big)/\sigma_Y-\widehat{\rho}\alpha_{(r:n)}$ for 
$r=1,2,\ldots,[ns]$.
Statistical properties of the residuals, which are analogous to those in the classical linear model, 
are summarized as follows.  	
\begin{theorem} Under the assumptions of Theorem 2, the following statements are true for the residuals:
(i) ${\rm E} \big( e_{[r:n]}^* \big)=0$; (ii) \begin{bf}
\begin{eqnarray}
{\rm var} \big( e_{[r:n]}^* \big) &=&\Big\{ \rho^2 \beta_{(rr:n)} + \big( 1 -\rho^2
 \big) \Big\}
+ \alpha_{(r:n)}^2 \frac{1}{n\sigma_Y^2} \frac{\Psi_n^{\rm I}(s)}{\Phi_n^2(s)}
  -2 \frac{1}{\sum_{r=1}^{[ns]}\alpha_{(r:n)}^2} \nonumber\\
&&\times 
\left\{  \rho^2\sum_{w=1}^{[ns]}\alpha_{(w:n)}\alpha_{(r:n)}\beta_{(rw:n)}+\alpha_{(r:n)}^2(1-\rho^2)\right\}\nonumber
\end{eqnarray}  \end{bf} (iii) 
${\rm E} \big( e_{[r:n]}^* \alpha_{(r:n)} \big)=0$; and (iv) ${\rm E} \big( e_{[r:n]}^* \widehat{Y}^*_{[r:n]}
 \big)=0$, 
where $\widehat{Y}^*_{[r:n]}=\mu_Y- \widehat{\rho} \alpha_{(r:n)}$. 
\end{theorem} 
\noindent The proof of Theorem 3 requires only simple algebra and is
thus omitted here. The theorem states 
that the residuals have mean $0$ and finite variance, and also states that they are uncorrelated with 
the scores $\alpha_{(r:n)}$ and the predicted values $\widehat{Y}_{[r:n]}$. Thus, the residual plots, which 
are the plots of (i) $r$ versus $e_{[r:n]}^*$, (ii) $\alpha_{(r:n)}$ versus $e_{[r:n]}^*$, and (iii)
 $\widehat{Y}_{[r:n]}$ versus $e_{[r:n]}^*$, have the same interpretations as those of the 
classical linear model. We plug in $\mu_Y$ and $\sigma_Y$ with their empirical estimators and use 
$e_{[r:n]}=\big(Y_{[r:n]}-\widehat{\mu}_Y-\big)/\widehat{\sigma}_Y-\widehat{\rho} \alpha_{(r:n)}$.

The residual sum of squares may be another useful tool for measuring the goodness of fit of 
the proposed model, as in the classical linear model. The residual sum of squares in our model is defined
 as 
\begin{equation} \nonumber 
{\rm RSS} = \sum_{r=1}^{[ns]} \bigg(\frac{Y_{[r:n]}-\widehat{\mu}_Y}{\widehat{\sigma}_Y} -\widehat{Y}_{[r:n]} \bigg)^2
\end{equation}
and will be used along with the residual plots as a guide for selecting a better score function.

Finally,  the proposed least-squares estimator (\ref{eqn:lse}) 
assumes that the regression line between $\alpha_{(r:n)}$ and $\big(Y_{[r:n]} - \widehat{\mu}_Y \big)$ 
has an intercept (at the $y$ axis) of $0$. Thus, if the model (or the score function) is correctly 
specified, then the intercept estimated by the regression (with intercept) should be close to $0$, and 
the estimated intercept therefore serves as a measure for checking the correctness of the score function. Note 
that the regression (without intercept) performed in this paper is based on 
observations of the top $[ns]$ ranks and assumes that the function passes through the origin (see Figure 4).

\subsection{An Estimator with Unranked Observations}
   
The least-squares estimator presented in Section 3.2 does not fully use the information contained in 
$\big\{ Y_{[r:n]}:=Y_i {\rm I}(R_i=r), r >m \big\}$;  
it is used only to estimate $\mu_Y$ and $\sigma_Y$, not to estimate $\rho$ itself. In this section, we briefly
demonstrate how $\widehat{\rho}$ can be modified to incorporate these unranked observations.

We consider the following modified estimator: 
\begin{equation} \nonumber 
\widehat{\rho}_{\rm m} 
\big(s \big) \equiv \frac{1}{\widehat{\sigma}_Y} \cdot
 \frac{\sum_{r=1}^{[ns]} \alpha_{(r:n)}\big\{  Y_{[r:n]} -\widehat{\mu}_Y  \big\} + 
\big(n-[ns] \big) \overline{\alpha}_{[ns]+} \big(\overline{Y}_{[ns]+} -\widehat{\mu} \big)}
{\sum_{r=1}^{[ns]} \alpha_{(r:n)}^2 + \big(n-[ns]\big) \overline{\alpha}_{[ns]+}^2},
\end{equation} 
where  $\overline{\alpha}_{[ns]+}=\sum_{r=[ns]+1}^n \alpha_{(r:n)} \big/ \big(n-[ns]\big)$ 
and $\overline{Y}_{[ns]+}=\sum_{r=[ns]+1}^n Y_{[r:n]}\big/ \big(n-[ns]\big)$. This modified
estimator also asymptotically approaches the normal distribution. Specifically, suppose that 
\begin{equation} \nonumber 
\widetilde{\alpha}_{(r:n)}=
\left\{
\begin{array}{ll}
\alpha_{(r:n)} & \quad r=1,2,\ldots,[ns],\\
& \\
\overline{\alpha}_{[ns]+} & \quad r=[ns]+1,[ns]+2,\ldots,n.
\end{array}
\right.
\end{equation}
We also suppose that \\$\widetilde{\Psi}_n^{\rm I} (s)= \big(1\big/n\big) \left\{ \sum_{r=1}^{n} \widetilde{\alpha}_{(r:n)}^2
 \sigma_{(r:n)}^2\right\}$, $\widetilde{\Psi}_n^{\rm II} (s)=\big(1\big/n\big)\left\{
 \sum_{r1=1}^{n} \sum_{r2=1}^{n} \widetilde{\alpha}_{(r1:n)} \widetilde{\alpha}_{(r2:n)}\beta_{(r1,r2:n)}^2
 \right\}$ and 
$\widetilde{\Phi}_n (s)= \big(1 \big/n \big) \sum_{r=1}^{n} \widetilde{\alpha}_{(r:n)}^2$. 
As in the previous section, $(1/n)$-scaled 
limits of $\widetilde{\Psi}_n^{\rm I} (s)$, $\widetilde{\Psi}_n^{\rm II} (s)$ and $\widetilde{\Phi}_n (s)$ exist; let 
 the limits be $\widetilde{\Psi}_{\infty}^{\rm I}(s)= \lim_{n \rightarrow \infty}\widetilde{\Psi}^{\rm I}_n (s)\big/n$, $\widetilde{\Psi}_{\infty}^{\rm II}(s)= \lim_{n \rightarrow \infty}\tilde{\Psi}^{\rm II}_n (s)\big/n$ and 
$\tilde{\Phi}_{\infty}(s)= \lim_{n \rightarrow \infty} \tilde{\Phi}_n (s) \big/n$, respectively. 
Then, we can write the following theorem.  
\begin{theorem}
Under 
the same assumptions as those of Theorem 2, the distribution of $\sqrt{n} \big(\widehat{\rho}_{\rm m} (s) - \rho \big)$ converges 
to the normal distribution with mean $0$ and variance 
 $\big(\widetilde{\Psi}_{\infty}^{\rm I}(s)/\sigma_Y^2+\rho^2\widetilde{\Psi}_{\infty}^{\rm II}(s)\big)/\widetilde{\Phi}_{\infty}^2(s)$
\end{theorem} 
\begin{proof} 
\begin{eqnarray}
\sqrt{n} \big( \widehat{\rho}_{\rm m} - \rho \big) &=& 
\sqrt{n} \bigg\{\frac{1}{\sum_{r=1}^{[ns]} \alpha_{(r:n)}^2 + \big(n-[ns]\big) \overline{\alpha}_{[ns]+}^2}\times
 \nonumber\\
&& \qquad \bigg( \sum_{r=1}^{[ns]} \alpha_{(r:n)} \Big( Y_{[r:n]} - \widehat{\mu}_Y \big)+\big(n-[ns] 
\big) \overline{\alpha}_{[ns]+} \big(\overline{Y}_{[ns]+}-\widehat{\mu}_Y \big) \bigg)-\rho\bigg\} \nonumber\\
&=& \sqrt{n} \bigg(\frac{1}{\widehat{\sigma}_Y} \frac{\sum_{r=1}^n \tilde{\alpha}_{(r:n)} \big( Y_{[r:n]} -
 \widehat{\mu}_Y \big)} {\sum_{r=1}^n \widetilde{\alpha}_{(r:n)}^2}-\rho\bigg), \nonumber 
\end{eqnarray}
the distribution of which converges to the normal distribution with mean $0$ and variance \\
$\big(\widetilde{\Psi}_{\infty}^{\rm I}(s)/\sigma_Y^2+\rho^2\widetilde{\Psi}_{\infty}^{\rm II}(s)\big)/\widetilde{\Phi}_{\infty}^2(s)$
following the same arguments presented in the proof of Theorem 2. 
\end{proof}

\section{Numerical Study}

In this section, we numerically investigate the advantage 
we can gain by choosing the correct score function to estimate 
$\rho={\rm corr} (Y,X)$. The performance of an estimator is 
measured in terms of its bias and its mean square error (MSE), which 
we numerically estimate based on 
$1000$ simulated data sets and the estimators obtained therefrom.

The data sets are generated from the regression model 
$Y_i = \beta_0 + \beta_1 X_i + \epsilon_i,$ $i=1,2,\ldots,n,$ where the $\epsilon_i$ are
independently drawn from $N(0,1)$. We consider three distributions for $X$: 
the uniform distribution on $[0,1]$, the standard normal distribution, and the gamma 
distribution with mean $1$ and variance $1/3$.  As stated in Section 2, the score function of 
the uniform distribution is almost equivalent to the identity score function $S_n(r)=r$. However, the normal distribution and the gamma distribution have 
heavier tails than does the uniform distribution, and their score functions are convex in the right 
tail. We set the parameters $\delta$ to ensure that $\rho = 0$, $0.3$, $0.5$ and $0.7$, where
 $\rho=\delta/\sigma_Y$. Finally, in each considered case, the sample size $n$ and 
the number of partially observed ranks $m$ are set to all possible combinations of 
 $n = 500$ or $2000$ and $r=20$, $50$, or $100$.
When estimating $\rho$, we apply
four different scores, including the proposed MOS-based score functions
obtained from the three distributions listed above and the identity score function, which is commonly used in practice. The approximated bias and MSE values are reported in Tables 1 and 2.

We can observe several interesting findings from these tables.
First, the correctly specified score
function performs better than do others when there exists a strong correlation between $X$ and $Y$ (when 
$\rho$ is large). However, when $\rho=0$, there is almost no difference among the four considered 
scores. Second, as the number of observations increases, in the sense that either $r$ or $n$ increases, 
the superiority of the correctly specified scores with respect to the others becomes apparent even when $\rho$ is not large. 
Third, as conjectured in the previous section, the scores based on the uniform distribution 
perform almost identically to the identity scores. Finally, the
differences between the correctly specified scores and the others are
significant regardless of $\rho$ or the sample size ($r$ or $n$) 
when the distribution of $X$ has a heavier right tail (the gamma distribution).

\begin{table}[htp!]
\begin{center} \scriptsize{
	\begin{tabular}{|c|c|c|c|c|c|c|c|c| }
        \hline
       &  &&\multicolumn{2}{|c|}{$U(0,1)$}&\multicolumn{2}{|c|}{$N(0,1)$} & \multicolumn{2}{|c|}{$G(3,3)$}\\
    \cline{2-9}
&	$\rho$&Dist  & Bias & MSE & Bias & MSE & Bias   & MSE  \\
         \hline
\multirow{15}{1cm}{r=20}  &  $0.0$&1:N&    0.0009  & 0.0167     & -0.0009  & 0.0165   & -0.0025  & 0.0184    \\
   &       	 &U&-$\ul{0.0009}$  &$\ul{0.0175}$     & -0.0027  & 0.0174  & 0.0051  & 0.0170   \\
&	 &N&    -0.0032  & 0.0103     &$\ul{0.0009}$  &$\ul{0.0104}$   & 0.0008  & 0.0097   \\
&	 &G&    0.0001  &${\bf 0.0059}$     & -0.0044  &${\bf 0.0059}$ &-$\ul{0.0021}$  &${\bf\ul{0.0057}}$  \\
                 \cline{2-9}
  &       $0.3$&1:N&    -0.0003  & 0.0161   & 0.0923  & 0.0243  & 0.2090  & 0.0600  \\
&		   &U&$\ul{0.0026}$  &$\ul{0.0158}$     & 0.0900  & 0.0256 & 0.2122  & 0.0603  \\
&		   &N&    -0.0715  &${\bf 0.0141}$     &-$\ul{0.0010}$  &${\bf \ul{0.0098}}$  & 0.0990  & 0.0194\\
&		   &G&    -0.1287  & 0.0218     & -0.0794  & 0.0118  &-$\ul{0.0024}$  &${\bf\ul{0.0057}}$  \\
		         \cline{2-9}
&		 $0.5$&1:N&    0.0012  & 0.0127     & 0.1437  & 0.0336  & 0.3430  & 0.1327       \\
&		   &U&$\ul{0.0035}$  &${\bf\ul{0.0121}}$     & 0.1474  & 0.0351  & 0.3476  & 0.1354       \\
&		   &N&    -0.1217  & 0.0222     &$\ul{0.0009}$  &${\bf\ul{0.0078}}$  & 0.1566  & 0.0330       \\
&		   &G&    -0.2180  & 0.0518     & -0.1241  & 0.0197   &-$\ul{0.0022}$  &${\bf\ul{0.0052}}$    \\
		        \cline{2-9}	
&		 $0.7$&1:N&    -0.0025  &${\bf 0.0087}$  & 0.2057  & 0.0522    & 0.4916  & 0.2535       \\
&		   &U&-$\ul{0.0003}$  &$\ul{0.0090}$  & 0.2053  & 0.0517   & 0.4886  & 0.2505       \\
&		   &N&    -0.1693  & 0.0334       &-$\ul{0.0009}$  &${\bf\ul{0.0057}}$  & 0.2266  & 0.0584       \\
&		   &G&    -0.3051  & 0.0958       & -0.1750  & 0.0338  &-$\ul{0.0057}$  &${\bf\ul{0.0040}}$     \\
\hline
\multirow{15}{1cm}{r=50}&	
 $0.0$&1:N&    0.0011  & 0.0077  & 0.0007  & 0.0076 & -0.0020  & 0.0069   \\
		&   &U&-$\ul{0.0019}$  &$\ul{0.0075}$  & -0.0005  & 0.0071 & 0.0034  & 0.0076  \\
		 &  &N&    -0.0038  & 0.0056  &$\ul{0.0008}$  &$\ul{0.0054}$ & 0.0008  & 0.0053  \\
		  & &G&    -0.0016  &${\bf 0.0035}$  & 0.0004  &${\bf 0.0039}$ &$\ul{0.0004}$  &${\bf\ul{0.0034}}$  \\
                 \cline{2-9}
 & $0.3$&1:N&    -0.0033  &${\bf 0.0064}$  & 0.0405  & 0.0083 &0.1116 & 0.0196 \\
&		   &U&-$\ul{0.0002}$  &$\ul{0.0066}$  & 0.0427  & 0.0082  & 0.1112  & 0.0187  \\
&	       &N&    -0.0418  & 0.0067  &-$\ul{0.0025}$  &${\bf\ul{0.0051}}$ & 0.0735  & 0.0110  \\
&		   &G&    -0.0988  & 0.0131  & -0.0604  & 0.0066 &-$\ul{0.0029}$  &${\bf\ul{0.0037}}$  \\
		        \cline{2-9}
&		 $0.5$&1:N&    0.0022  & 0.0050    & 0.0657  & 0.0096 & 0.1878  & 0.0413    \\
&		   &U&-$\ul{0.0011}$  &${\bf\ul{0.0049}}$  & 0.0687  & 0.0101 & 0.1882  & 0.0408    \\
&		   &N&    -0.0717  & 0.0093    &$\ul{0.0022}$  &${\bf\ul{0.0040}}$ & 0.1217  & 0.0192    \\
&		   &G&    -0.1652  & 0.0298    & -0.0982  & 0.0123 &-$\ul{0.0023}$  &${\bf\ul{0.0031}}$  \\
		         \cline{2-9}		
&		 $0.7$&1:N&    -0.0006  & 0.0035    & 0.0897  & 0.0116 & 0.2655  & 0.0744    \\
&		   &U&$\ul{0.0006}$  &${\bf\ul{0.0035}}$  & 0.0975  & 0.0133 & 0.2641  & 0.0736    \\
&		   &N&    -0.0998  & 0.0127    &$\ul{0.0044}$  &${\bf\ul{0.0026}}$ & 0.1655  & 0.0307    \\
&		   &G&    -0.2293  & 0.0544    & -0.1401  & 0.0216 &$\ul{0.0040}$  &${\bf\ul{0.0021}}$ \\
\hline
\multirow{15}{1cm}{r=100}&   $0.0$&1:N&    0.0005  & 0.0042     & 0.0005  & 0.0042  & 0.0040  & 0.0042     \\
	 &	   &U&-$\ul{0.0017}$  &$\ul{0.0040}$     & 0.0000  & 0.0043 & -0.0021  & 0.0039     \\
	&	   &N&    0.0003  & 0.0035     &-$\ul{0.0029}$  &$\ul{0.0036}$ & -0.0012  & 0.0040     \\
	&	   &G&    0.0003  &${\bf 0.0028}$     & 0.0005  &${\bf 0.0027}$ &$\ul{0.0007}$  &${\bf\ul{0.0027}}$     \\
                \cline{2-9}
   &      $0.3$&1:N&    0.0013  & 0.0035  & 0.0107  & 0.0034 & 0.0490  & 0.0061    \\
	&	   &U&$\ul{0.0008}$  &$\ul{0.0037}$  &0.0099  &0.0036 & 0.0498  & 0.0062    \\
	&	   &N&    -0.0210  &${\bf 0.0033}$  &$\ul{ 0.0005}$  &$\ul{ 0.0034}$ & 0.0517  & 0.0063    \\
&		   &G&    -0.0714  & 0.0074  & -0.0479  &${\bf 0.0049}$ &-$\ul{0.0038}$  &${\bf\ul{ 0.0026}}$ \\
		        \cline{2-9}
 &		 $0.5$&1:N&    0.0003  & 0.0028    & 0.0161  & 0.0031 & 0.0896  & 0.0110    \\
&		   &U&-$\ul{0.0002}$  &${\bf\ul{ 0.0028}}$  & 0.0152  & 0.0031 & 0.0890  & 0.0108    \\
&		   &N&    -0.0352  & 0.0036    &-$\ul{0.0011}$  &${\bf\ul{ 0.0025}}$ & 0.0836  & 0.0096    \\
&		   &G&    -0.1187  & 0.0159    & -0.0823  & 0.0088 &-$\ul{0.0035}$  &${\bf\ul{ 0.0020}}$  \\
		        \cline{2-9}		
&		 $0.7$&1:N&    0.0008  & 0.0017    & 0.0230  & 0.0022 & 0.1210  & 0.0166    \\
&		   &U&$\ul{    0.0004}$  &${\bf\ul{ 0.0017}}$  & 0.0225  & 0.0023 & 0.1247  & 0.0175    \\
&		   &N&    -0.0468  & 0.0038    &-$\ul{0.0002}$  &${\bf\ul{ 0.0017}}$ & 0.1186  & 0.0158    \\
&		   &G&    -0.1689  & 0.0297    & -0.1140  & 0.0143 &-$\ul{0.0016}$  &${\bf\ul{ 0.0014}}$  \\ \hline
     \end{tabular}
       }
\bigskip
	   \caption{$n=500$: In the MSE columns, the numbers in bold-faced text are the smallest among the evaluated score functions. In both the bias and MSE columns, the underlined numbers are the true values (those from the correctly specified score functions).}
\end{center}
\end{table}

\begin{table}[htp!]
\begin{center} \scriptsize{
	\begin{tabular}{|c|c|c|c|c|c|c|c|c| }
        \hline
        & &&\multicolumn{2}{|c|}{$U(0,1)$}&\multicolumn{2}{|c|}{$N(0,1)$}& \multicolumn{2}{|c|}{$G(3,3)$}\\
          \cline{2-9}
&		$\rho$&Dist &  Bias & MSE & Bias & MSE & Bias   & MSE  \\
         \hline
   \multirow{15}{1cm}{r=20}&       $0.0$&1:N&    -0.0007  & 0.0173  & 0.0000  & 0.0168 & 0.0006  & 0.0165  \\
        &   &U&$\ul{    0.0077}$  &$\ul{ 0.0168}$  & 0.0029  & 0.0166 & 0.0066  & 0.0170  \\
	&	   &N&    -0.0042  & 0.0072  &-$\ul{0.0023}$  &$\ul{ 0.0069}$ & 0.0008  & 0.0068  \\
	&	   &G&    -0.0014  &${\bf 0.0032}$  & 0.0017  &${\bf 0.0034}$ &$\ul{ 0.0005}$  &${\bf\ul{ 0.0031}}$  \\
        \cline{2-9}
       &  $0.3$&1:N&    -0.0020  &${\bf 0.0157}$  & 0.1650  & 0.0433 & 0.3691  & 0.1519    \\
	&	   &U&-$\ul{0.0033}$  &$\ul{ 0.0158}$  & 0.1630  & 0.0427 & 0.3707  & 0.1537    \\
	&	   &N&    -0.1068  & 0.0174  &$\ul{ 0.0028}$  &${\bf\ul{ 0.0062}}$ & 0.1319  & 0.0249    \\
	&	   &G&    -0.1680  & 0.0312  & -0.0925  & 0.0118 &$\ul{ 0.0041}$  &${\bf\ul{ 0.0031}}$  \\
	\cline{2-9}
	&	 $0.5$&1:N&    -0.0053  & 0.0123  & 0.2821  & 0.0931  & 0.6109  & 0.3891    \\
	&	   &U&$\ul{    0.0009}$  &${\bf\ul{ 0.0122}}$  & 0.2714 & 0.0870 & 0.6163 &	0.3952  \\
	&	   &N&    -0.1822  & 0.0387  &-$\ul{0.0010}$  &${\bf\ul{ 0.0056}}$ & 0.2238  & 0.0566    \\
&		   &G&    -0.2821  & 0.0821  & -0.1537  & 0.0263 &$\ul{ 0.0020}$  &${\bf\ul{ 0.0031}}$  \\
		\cline{2-9}	
	&	 $0.7$&1:N&    -0.0013  &${\bf 0.0083}$    & 0.3825  & 0.1558 & 0.8650  & 0.7618    \\
&		   &U&$\ul{    0.0030}$  &$\ul{ 0.0091}$  & 0.3859  & 0.1586 & 0.8623  & 0.7575    \\
&		   &N&    -0.2528  & 0.0671    &$\ul{ 0.0014}$  &${\bf\ul{ 0.0039}}$ & 0.3113  & 0.1026    \\
&		   &G&    -0.3936  & 0.1567    & -0.2182  & 0.0496  &$\ul{ 0.0037}$  &${\bf\ul{ 0.0029}}$   \\ \hline
 \multirow{15}{1cm}{r=50}&   $0.0$&1:N&    -0.0013  & 0.0071  & 0.0012  & 0.0070 & 0.0009  & 0.0071  \\
		&   &U&$\ul{    0.0024}$  &$\ul{ 0.0071}$  & 0.0020  & 0.0066 & -0.0008  & 0.0075  \\
		&   &N&    0.0015  & 0.0036  &$\ul{ 0.0040}$  &$\ul{ 0.0038}$ & -0.0007  & 0.0033  \\
		&   &G&    0.0002  &${\bf 0.0018}$  & -0.0009  &${\bf 0.0018}$ &$\ul{ 0.0011}$  &${\bf\ul{ 0.0019}}$  \\
        \cline{2-9}
    &     $0.3$&1:N&    0.0016  & 0.0062  & 0.1147  & 0.0196 & 0.2589  & 0.0737    \\
	&	   &U&-$\ul{0.0028}$  &${\bf\ul{ 0.0062}}$  & 0.1136  & 0.0192 & 0.2667  & 0.0776    \\
	 &      &N&    -0.0851  & 0.0102  &$\ul{ 0.0012}$  &${\bf\ul{ 0.0035}}$ & 0.1121  & 0.0159    \\
&		   &G&    -0.1456  & 0.0229  & -0.0818  & 0.0083 &$\ul{ 0.0002}$  & ${\bf\ul{0.0018}}$  \\
	\cline{2-9}
 &		 $0.5$&1:N&    0.0025  & 0.0051    & 0.1953  & 0.0436 & 0.4396  & 0.1988    \\
&		   &U&-$\ul{0.0008}$  &${\bf\ul{ 0.0048}}$  & 0.1928  & 0.0426 & 0.4389  & 0.1986    \\
&		   &N&    -0.1486  & 0.0246    &-$\ul{0.0006}$  &${\bf\ul{ 0.0025}}$ & 0.1854  & 0.0376    \\
&		   &G&    -0.2457  & 0.0617    & -0.1371  & 0.0202 &-$\ul{0.0022}$  &${\bf\ul{ 0.0017}}$  \\
		\cline{2-9}	
&		 $0.7$&1:N&    0.0006  & 0.0035    & 0.2698  & 0.0767 & 0.6124  & 0.3806    \\
&		   &U&$\ul{    0.0011}$  &${\bf\ul{ 0.0034}}$  & 0.2691  & 0.0761 & 0.6195  & 0.3882    \\
&		   &N&    -0.2019  & 0.0426    &$\ul{ 0.0015}$  &${\bf\ul{ 0.0020}}$ & 0.2578  & 0.0690    \\
&		   &G&    -0.3429  & 0.1184    & -0.1901  & 0.0372 &-$\ul{0.0018}$  &${\bf\ul{ 0.0014}}$  \\ \hline
\multirow{15}{1cm}{r=100}&  $0.0$&1:N&    0.0002  & 0.0036  & 0.0015  & 0.0038 & 0.0039  & 0.0034  \\
		  & &U&-$\ul{0.0008}$  &$\ul{ 0.0035}$  & 0.0013  & 0.0036 & -0.0002  & 0.0037  \\
		  & &N&    0.0012  & 0.0021  &$\ul{ 0.0039}$  &$\ul{ 0.0022}$ & -0.0029  & 0.0021  \\
		  & &G&    -0.0003  &${\bf 0.0013}$  & -0.0014  &${\bf 0.0013}$ &$\ul{ 0.0001}$  &${\bf\ul{ 0.0013}}$  \\
        \cline{2-9}
        & $0.3$&1:N&    0.0031  & 0.0032    & 0.0791  & 0.0095 & 0.1863  & 0.0380    \\
	&	   &U&-$\ul{0.0006}$  &${\bf\ul{ 0.0029}}$  & 0.0764  & 0.0089 & 0.1847  & 0.0374    \\
	&	   &N&    -0.0698  & 0.0067    &$\ul{ 0.0000}$  &${\bf\ul{ 0.0021}}$ & 0.0920  & 0.0105    \\
	&	   &G&    -0.1247  & 0.0167    & -0.0711  & 0.0062 &-$\ul{0.0006}$  &${\bf\ul{ 0.0012}}$  \\
	\cline{2-9}
	&	 $0.5$&1:N&    -0.0008  & ${\bf 0.0024}$    & 0.1302  & 0.0196 & 0.3124  & 0.1006    \\
	&	   &U&$\ul{    0.0001}$  &$\ul{ 0.0029}$  & 0.1284  & 0.0191 & 0.3114  & 0.1000    \\
	&	   &N&    -0.1107  & 0.0139    &$\ul{ 0.0009}$  &${\bf\ul{ 0.0016}}$ & 0.1530  & 0.0252    \\
&		   &G&    -0.2105  & 0.0453    & -0.1204  & 0.0155  &-$\ul{0.0001}$  &${\bf\ul{ 0.0011}}$  \\
	\cline{2-9}		
&		 $0.7$&1:N&    -0.0012  & 0.0018    & 0.1793  & 0.0339 & 0.4344  & 0.1910    \\
&		   &U& $\ul{   0.0001}$  &${\bf\ul{ 0.0017}}$  & 0.1793  & 0.0340 & 0.4365  & 0.1927    \\
&		   &N&    -0.1551  & 0.0251    &$\ul{ 0.0010}$  &${\bf\ul{ 0.0011}}$ & 0.2156  & 0.0479    \\
&		   &G&    -0.2923  & 0.0861    & -0.1691  & 0.0292 &-$\ul{0.0009}$  &${\bf\ul{ 0.0009}}$  \\ \hline 
       \end{tabular}}
\bigskip
	   \caption{$n=2000$: In the MSE columns, the numbers in bold-faced text are the smallest among the evaluated score functions. In both the bias and MSE columns, the underlined numbers are the true values (those from the correctly specified score functions).}
\end{center}
\end{table}

\section{Data Examples}

\subsection{Data Description}

To investigate how the attention of investors affects stock returns, we merge 
the hand-collected \texttt{Daum.net} rank data set and the financial data from \texttt{FnGuide}. 
 We illustrate how the returns of attention-grabbing stocks  
fluctuate around the event dates when investors pay attention to these stocks.  
The variables to be used in the analysis are as follows.
(1)  ``R":  The rank of an individual stock on day $t$; if the rank value is $1$, then the stock is 
the most frequently discussed stock on the \texttt{Daum} stock message board on that day. This is the key variable that 
measures the degree of investor attention. 
(2) ``RN": Raw returns on day $t+1$ (the next day) (\%), which is of
primary interest and is the quantity that we wish to predict.
(3)`` R0": Raw returns on day $t$ (\%).  
(4) ``R1": Raw returns on day $t-1$ (\%). 
(5) ``R2": Raw returns on day $t-2$ (\%). 
(6) ``R3": Raw returns on day $t-3$ (\%). 
(7) ``R4": Raw returns on day $t-4$ (\%). 
(8) ``R5": Raw returns on day $t-5$ (\%).
(9) ``ME": Market capitalization (1 trillion Korean won). 
(10) ``T": Turnover ratio defined as the trading volume divided by the number of outstanding shares. 
(11) ``TA": Turnover ratio defined as the trading volume divided by market capitalization.

\subsection{Attention and Predictive Stock Returns} 

As stated previously, the primary goal of our analysis is to determine how the returns of attention-grabbing stocks 
will fluctuate around the event dates when investors pay attention to these stocks. 
The next-day return can also be influenced by several other factors 
in addition to investor attention. To account for the effects of these other factors, we consider 
the residuals obtained after regressing the next-day return against all other covariates except the rank, ``R". 
These residuals are obtained from the multiple linear regression model, which is defined as follows:
\begin{equation} \label{eqn:reg}
{\rm RN}_i= \beta_0 + \sum_{l=0}^5 \beta_{l+1} {\rm R}l_i + \beta_7 {\rm ME}_i  + \beta_8 {\rm T}_i 
+\beta_9 {\rm TA}_i  + \epsilon_i, \quad i=1,2,\ldots,n,  
\end{equation}
where $n (=1,771)$ is the total number of companies on the market. Let $Y_{i}^t$ be the absolute 
(value of the) residual of company $i$ obtained from the regression (\ref{eqn:reg}).
We then select the absolute residuals 
whose ranks are reported to be within the top 30 for the primary analysis. Below, 
$Y_{[r:n]}^t$ is the absolute residual
 corresponding to rank $r$ on day $t$ for $r=1,2\ldots,30$ and $t=1,2,\ldots,T (=537)$.

In Figure 1, we plot the quantiles of $\big\{Y_{[r:n]}^t, t=1,2,\ldots,T\big\}$ for 
each $r=1,2,\ldots,30$. This figure reveals 
that $Y_{[r:n]}^t$ is not increasing at $r=1$ and $2$, which we hypothesize reflects 
the heterogeneity of investor expectations with regard to highly 
attention-grabbing stocks. In other words, 
the ranking of the \texttt{Daum} board is purely determined by the attention 
of individual investors, and stocks related to news, 
that is difficult to characterize as either good or bad, often receive
 the greatest attention and the highest ranks. 
We introduce an additional term to explain 
this apparent local non-monotonicity, and 
consider the model
\begin{equation}
\label{eqn:model-realdata}
Y_{[r:n]}^t = \mu_Y^t + \rho^t \sigma_Y^t  \alpha_{(r:n)} + \gamma^t {\rm I}(r \leq 2) + \eta_{[r:n]}^t, ~~ r = 1, 2, \ldots, 30,
\end{equation}
for $t= 1,2, \ldots, T$ with $T=537$ and $n=1,771$.

\begin{figure}[htb!]
 \centering
\begin{minipage}{0.6\textwidth}
  \includegraphics[width=1.0\linewidth]{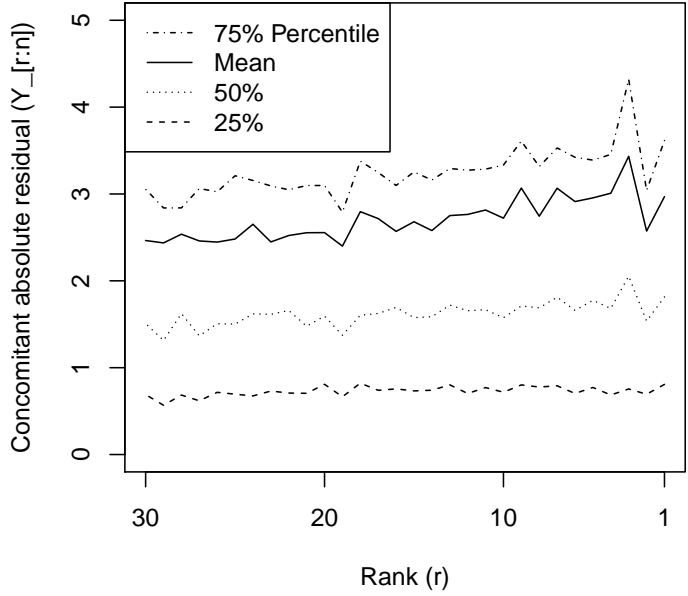}
\end{minipage}%
\caption{Plot of the means and quantiles of $\big\{Y_{[r:n]}^t, t=1,2,\ldots,T\big\}$ for 
each $r=1,2,\ldots,30$.}
\end{figure}

\subsection{Regression with Ranks} 

In the regression model, we consider the scores from the standardized 
distributions of the location-scale families generated by the following three 
distributions: (i) a uniform distribution on $(0,1)$ (called the uniform score), 
(ii) a positive normal distribution 
$X = |Z|$, $Z \sim N(0,1)$ (called the half-normal score), and (iii) a power-law 
distribution $X$ whose CDF is $F(x) = 1 - x^{-\alpha}$ with
 $\alpha = 2.3$ (called the power-law score). 
The scores are illustrated on different scales in Figure 2. 

\begin{figure}[htb!]
 \centering
\begin{minipage}{0.8\textwidth}
  \includegraphics[width=1.0\linewidth]{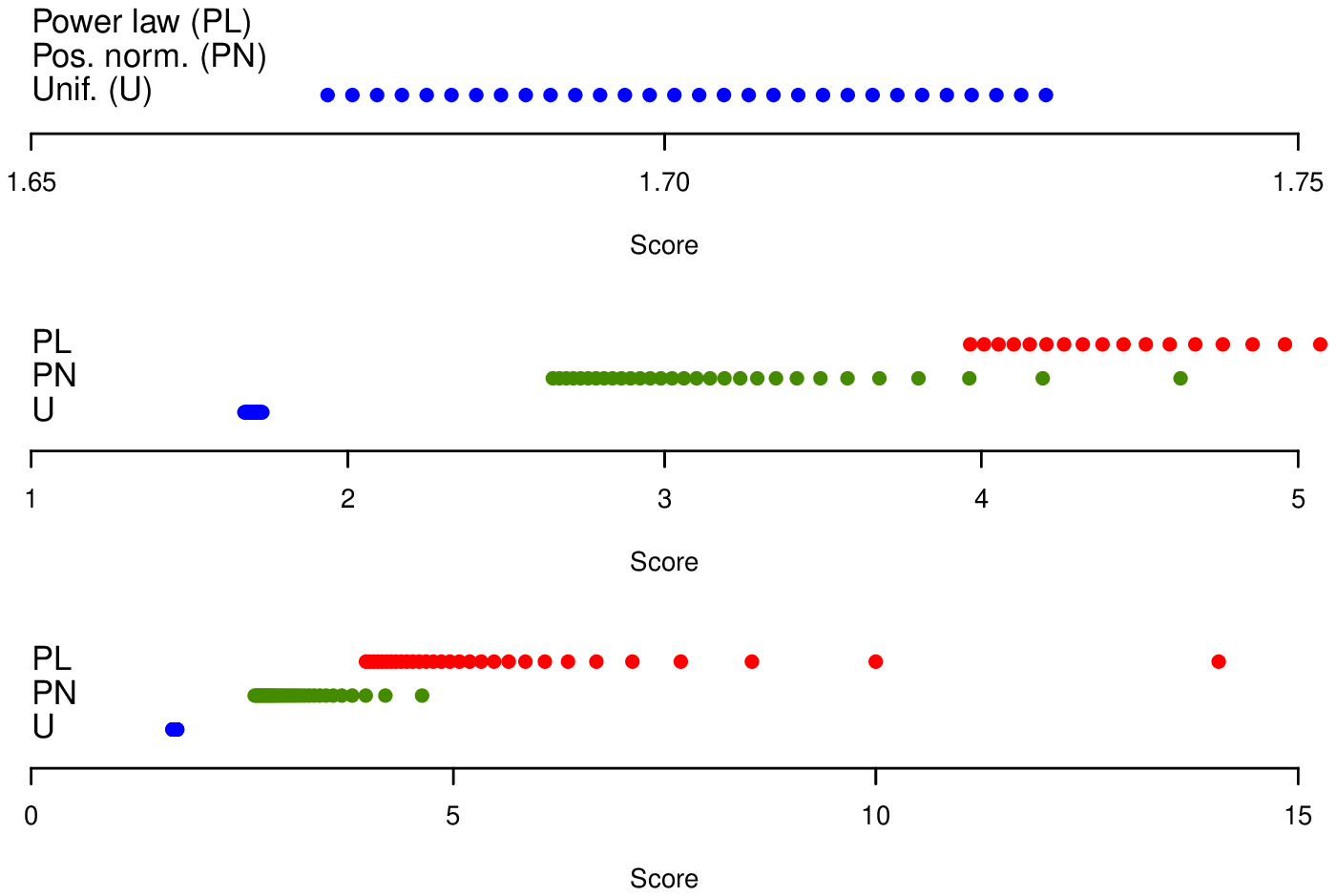}
\end{minipage}
\caption{$\{ \alpha_{(r:n)}\}_{r=1,\ldots,30}$ for each distribution on different scales, where $n = 1,771$.} 
\end{figure}

We estimate $\rho^t$ and $\gamma^t$ to minimize the empirical squared-error loss of the 
model (\ref{eqn:model-realdata}) by iterating the following steps:
\begin{enumerate} 
\item Given the least-squares estimator of $\rho$, denoted by $\widehat{\rho}_{(0)}^t$, update 
the estimate of $\gamma$ as follows: 
\begin{equation} \nonumber 
\widehat{\gamma}^t = \frac{1}{2} \left[
	\left( Y_{[1:n]}^t - \mu_Y^t - \sigma_Y^t \widehat{\rho}_{(0)}^t \alpha_{(1:n)} \right) +
	\left( Y_{[2:n]}^t - \mu_Y^t - \sigma_Y^t \widehat{\rho}_{(0)}^t \alpha_{(2:n)} \right)
	\right].
\end{equation} 
\item Given the estimate of $\gamma$, denoted by $\widehat{\gamma}_{(0)}^t$, update the estimate
of $\rho$ using the LSE proposed in the previous section as follows: 
\begin{equation} \nonumber 
\widehat{\rho}^t = \frac{1}{\sigma_Y^t} \left\{
	\frac{\sum_{r=1}^{30} \alpha_{(r:n)}
\big(Y_{[r:n]}^t - \mu_Y^t - \widehat{\gamma}_{(0)}^t {\rm I} (r \le 2) \big) }
{\sum_{r=1}^{30} \alpha_{(r:n)}^2 } \right\}.
\end{equation} 
\end{enumerate} 
In the analysis, the initial value $\widehat{\rho}_{(0)}^t$ is obtained from the preliminary linear regression on  $\left\{ \big( \alpha_{(r:n)}, (Y_{[r:n]}^t - \mu_Y^t)/\sigma_Y^t \big) \right\}_{r=3, \cdots, 30}$, $t=1,\ldots, T$, in which the data corresponding to $r=1,2$ are excluded.  By contrast, $\mu_Y^t$ and $\sigma_Y^t$ are 
estimated based on their empirical values as follows:
 $\widehat{\mu}_Y^t = \big(\sum_{r=1}^{n} Y_{[r:n]}^t\big) \big/n$ and 
$\big(\widehat{\sigma}_Y^t)^2 = \sum_{r=1}^{n} \big( Y_{[r:n]}^t - \widehat{\mu}_Y^t \big)^2 \big/n$.

To choose the most appropriate score function among the three considered, 
we follow the guidelines presented in Section 3.2
and perform a residual analysis. First, we plot $\alpha_{(r:n)}$ and the quantiles 
of the corresponding residuals to identify any remaining trend not explained by the model 
(see Figure 3). This figure
shows that the uniform score and the half-normal score exhibit additional linear trends 
not explained by the linear model (\ref{eqn:model-realdata}), whereas the 
power-law score performs well. Second, 
we plot  
\begin{equation} \nonumber 
\left(\alpha_{(r:n)}, \frac{Y_{[r:n]}^t - \widehat{\mu}_Y}{\widehat{\sigma}_Y} \right), ~~r=3,4,\ldots,30,
t=1,2,\ldots,T,
\end{equation}
and apply the least-squares fits with/without intercept. As we know from the model 
(\ref{eqn:con-linear}), the estimated regression line with intercept should cross the origin if the scores are correctly 
specified. Figure 4 reveals that the (estimate of) the intercept of the power-law score is closest to zero 
among the intercepts of the three considered scores. 
Finally, the residual sums of squares of the three scores are found to be 
$152186.7$, $150706.3$, and $150288.9$, respectively. This finding also 
supports the superiority of the power-law score function, and in the following analysis, we focus on
the power-law score function.

\begin{figure}[htb!]
 \centering
\begin{minipage}{0.7\textwidth}
  \includegraphics[width=1.0\linewidth]{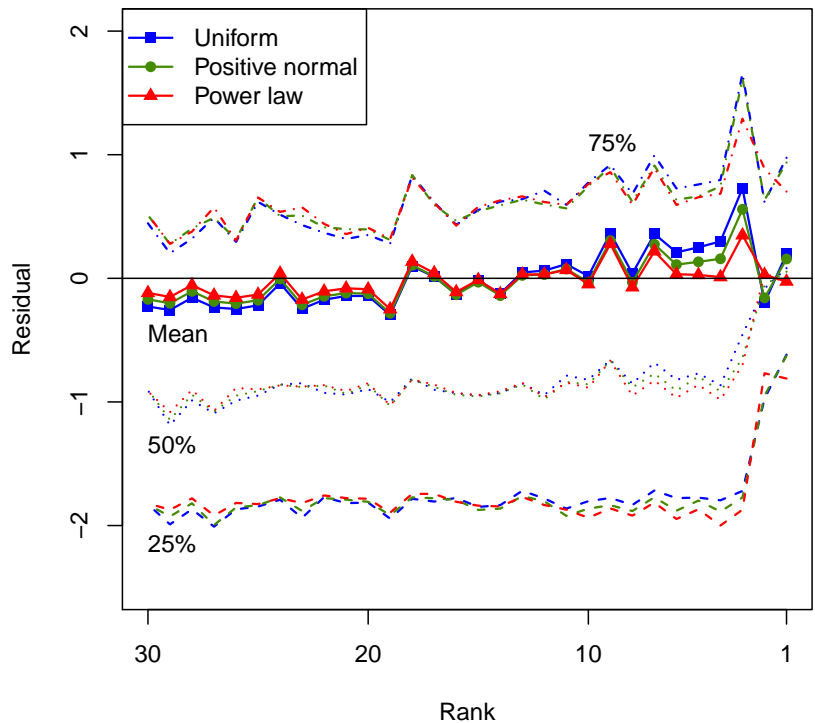}
\end{minipage} %
\caption{The averages and quantiles of the residuals for each rank.}
\end{figure}

\begin{figure}[htb!]
 \centering
\begin{minipage}{0.7\textwidth}
 \includegraphics[width=1.0\linewidth]{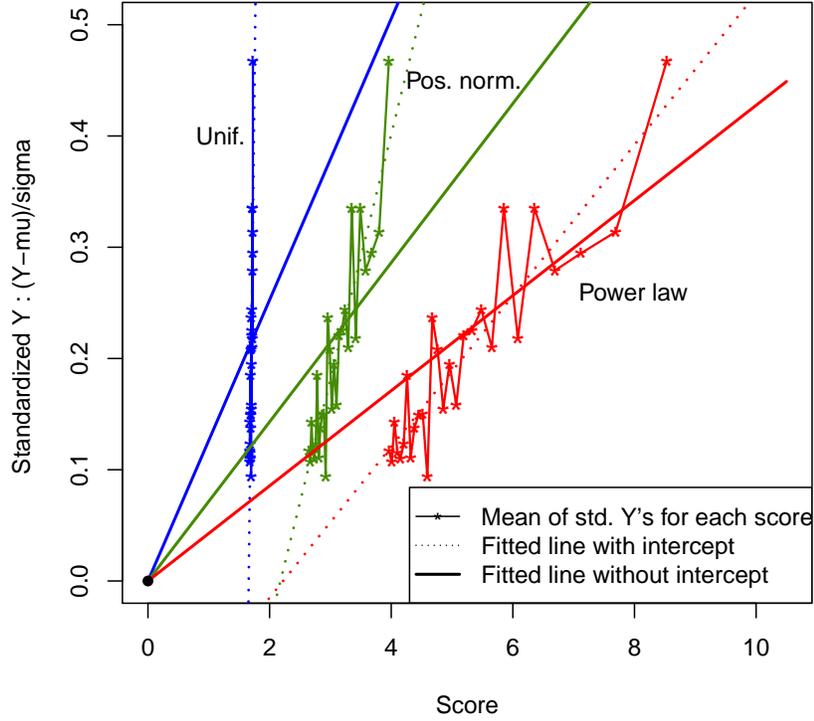}
\end{minipage}
\caption{Check of proportionality between the standardized residuals
and the scores. Points that are marked by `*' represent the average standardized residuals for each score (rank),
the dotted line represents the fitted model for a na\"{\i}ve simple regression with intercept, and the solid line represents our model.
Refer to Sections 3.5 and 5.3 for details.} 
\end{figure}

\subsection{Test of the Effect of Investor Attention on the Next-day Returns}

The primary goal of the analysis is to investigate whether the attention of investors affects the returns 
of a stock on the following day. Specifically, we are interested in testing $\mathcal{H}_0: \rho=0$ under 
the assumption that $\rho^t=\rho$ for every $t$. To test this hypothesis, we consider a combined statistic
of $\{ \widehat{\rho}^t, t=1,2,\ldots, T \}$, that is, 
\begin{equation} \label{eqn:teststat}
{\bf t}_{\rho} = \frac{1}{\sqrt{T}} \sum_{t=1}^T U_t,
\end{equation}
where $U_t= \sqrt{n} \widehat{\rho}^t$. Here, the estimates of $\rho$ for each day $t$, denoted
 by $\widehat{\rho}^t$, are serially dependent on each other, as are the $U_t$s. Thus, to obtain 
the reference distribution of ${\bf t}_{\rho}$, we further assume that
 $\{U_t, ~t=1,2,\ldots,T\}$ is stationary and that ${\rm E} | U_t |^{2+\kappa} < \infty$ for $\kappa>0$. 
Under these assumptions, 
the null distribution of ${\bf t}_{\rho}$ is asymptotically normal with mean $0$ and variance 
\begin{equation}\nonumber
\lim_{T \rightarrow \infty} \frac{1}{T} \sum_{k=0}^T \big( T-k\big) {\rm cov} (U_t, U_{t+k} ).
\end{equation}
The variance can be empirically estimated from the observed values of $\{U_t,~t=1,2,\ldots,T \}$ as
 \begin{equation}\nonumber
\frac{1}{T} \sum_{k=0}^m \big( T-k\big) \widehat{\rm cov} (U_t, U_{t+k} )
\end{equation}
for sufficiently large $m$, where $\widehat{\rm cov} (U_t, U_{t+k} )$ denotes
the empirical covariance of the observed statistics $(U_1, U_{1+k}), (U_2, U_{2+k}), \ldots, (U_{T-k+1} ,U_T)$.  An additional interesting feature of the combined procedure is that the test statistic ${\bf t}_{\rho}$ is a rough estimator of $\rho$ for all $T$ trading days (after 
the scaling). It is calculated as  
\begin{eqnarray}
\frac{1}{\sqrt{T}} \sum_{t=1}^T U_t &=& \frac{1}{\sqrt{T}} \sum_{t=1}^T \sqrt{n} \widehat{\rho}^t = 
\sqrt{n T} \left( \frac{1}{T} \sum_{t=1}^T \widehat{\rho}^t \right) \nonumber\\
&=& \sqrt{n T} \frac{1}{T}
\sum_{t=1}^T \frac{  \sum_{r=1}^{30} \alpha_{(r:n)} \big\{ Y_{[r:n]}^t - \widehat{\mu}_Y^t  - \widehat{\gamma}^t {\rm I} (r \le 2) \big\}} 
{\hat{\sigma}_Y^t \sum_{r=1}^{30} \alpha_{(r:n)}^2} \nonumber \\
&\approx& \sqrt{n T} \frac{1}{\widehat{\sigma}_Y}
\frac{ \sum_{t=1}^T \sum_{r=1}^{30} \alpha_{(r:n)} \big\{ Y_{[r:n]}^t - \widehat{\mu}_Y- \widehat{\gamma}^t {\rm I} (r \le 2)  \big\}} 
{T \sum_{r=1}^{30} \alpha_{(r:n)}^2} \approx \sqrt{ n T}  \widehat{\rho}^{\rm lse},  \label{eqn:equiv}
\end{eqnarray}
where $\widehat{\rho}^{\rm lse}$ is the least-squares estimator under the assumption that $\rho^t=\rho$
for all $t$. The difference between the right- and left-hand sides of (\ref{eqn:equiv}) lies in the definition of
 $\widehat{\gamma}^t$, which is defined using $\widehat{\rho}^t$ rather than 
$\widehat{\rho}^{\rm lse}$.

The results of the test indicate that the average value of $\widehat{\rho}^t$, which is an estimator 
of $\rho$, is $0.043$.
The $p$-value obtained when testing $\mathcal{H}_0:
\rho=0$ is less than $10^{-5}$ and statistically supports the association between investor
attention and the next-day returns of the stocks.

\section{Conclusion}

In this paper, we study a regression problem based on a partially observed rank covariate. We propose
a new set of score functions and study their application in simple linear regression. We demonstrate 
that the least-squares estimator that is calculated based on the newly proposed score consistently estimates the 
correlation coefficient between the response and the unobserved true covariate if the score 
function is correctly specified. We also define procedures based on the obtained residuals to identify 
the correct score function for the given data. The proposed estimator and procedures are 
applied to rank data collected from \texttt{Daum.net}, and we empirically verify the association between
investor attention and next-day stock returns.

We finally conclude the paper with two discussions on the proposed score function. First, 
the application of 
the proposed score function is not restricted to linear regression but may also be appropriate for other statistical 
procedures based on rank, including the well-known rank aggregation problem 
\citep{Breitling:2004,Eisinga:2013}. Second, the score function still can be used for 
the  the multiple linear regression model
\begin{equation} \nonumber
Y_i = X_i \beta + {\bf Z}_i^{\rm T} \eta+ \epsilon_{i}
\end{equation}
with an additional covariate vector ${\bf Z}=\big(Z_1,Z_2,\ldots,Z_q\big)^{\rm T}$. Similarly to the case of 
the simple linear regression, we have the representations
\begin{equation} \nonumber
Y_{[r:n]} = \mu_Y + \frac{X_{(r:n)}-\mu_X}{\sigma_X}  \delta+ \big(  {\bf Z}_{[r:n]} - \mu_{\mZ} \big)^{\rm T} \eta + \epsilon_{[r:n]},
\end{equation}
where  $\delta=\beta \sigma_X$ and $\epsilon_{[r:n]}$, $r=1,2,\ldots,[ns]$, have mean $0$ and independent to each other.
Again, the least-squares estimators of  
$\delta$ and $\eta$ are defined as the solutions to  
\begin{eqnarray}
&& \begin{pmatrix}
\begin{array}{ll}
\sum_{r=1}^{[ns]} \alpha_{(r:n)}^2 & \sum_{r=1}^{[ns]} \alpha_{(r:n)} \big( \mZ_{[r:n]}-\overline{\mZ} \big) \\
  \sum_{r=1}^{[ns]}  \big( \mZ_{[r:n]}-\overline{\mZ} \big)^{\rm T} \alpha_{(r:n)}
& \sum_{r=1}^{[ns]} \big( \mZ_{[r:n]}-\overline{\mZ} \big)^{\rm T} \big( \mZ_{[r:n]}-\overline{\mZ} \big)
\end{array}
\end{pmatrix}
\begin{pmatrix}
\begin{array}{l}
\widehat{\delta} \\
\widehat{\eta}
\end{array}
\end{pmatrix} \nonumber\\
&& \hskip5cm
=
\begin{pmatrix}
\begin{array}{l}
\sum_{r=1}^{[ns]} \alpha_{(r:n)} \big( Y_{[r:n]} - \overline{Y} \big) \\
\sum_{r=1}^{[ns]}  \big( \mZ_{[r:n]}-\overline{\mZ} \big)^{\rm T} \big( Y_{[r:n]} - \overline{Y} \big).
\end{array}
\end{pmatrix},  \nonumber
\end{eqnarray}
and conjecture that they consistently estimate $\delta$ and $\eta$.

\section*{Appendix}

\subsection*{A.1 Proof of Theorem 2}

Note that $\widehat{\sigma}_Y^2/\sigma_Y^2$ converges in probability to 1 as $n \rightarrow \infty$ and 
\begin{eqnarray}
\sqrt{n} \big( \widehat{\rho}(s) - \rho \big) &=&
\sqrt{n} \left\{ \frac{\sigma_Y}{\widehat{\sigma}_Y}\frac{1}{\sigma_Y}
\frac{\sum_{r=1}^{[ns]}  \alpha_{(r:n)}  \big( Y_{[r:n]} - \widehat{\mu}_Y \big)}
{\sum_{r=1}^{[ns]} \alpha_{(r:n)}^2} - \rho \right\}  \nonumber\
\end{eqnarray} 
has the same limiting distribution with 
\begin{eqnarray} 
&& \sqrt{n} \left\{ \frac{1}{\sigma_Y}
\frac{\sum_{r=1}^{[ns]}  \alpha_{(r:n)}  \big( Y_{[r:n]} - \widehat{\mu}_Y \big)} 
{\sum_{r=1}^{[ns]} \alpha_{(r:n)}^2} - \rho \right\} \nonumber\\
&& \qquad = \sqrt{n} \Bigg\{  \frac{1}{\sigma_Y}   \frac{ \sum_{r=1}^{[ns]}
\alpha_{(r:n)} \big(Y_{[r:n]} - m (X_{(r:n)}) \big)} { \sum_{r=1}^{[ns]} \alpha_{(r:n)}^2 } 
+ \rho \frac{\sum_{r=1}^{[ns]} \alpha_{(r:n)} \left(\left(\frac{X_{(r:n)}-\mu_X}{\sigma_X}\right) - \alpha_{(r:n)} \right)}{\sum_{r=1}^{[ns]} \alpha_{(r:n)}^2} \nonumber \\
&& \qquad \qquad \qquad \qquad \qquad \qquad \qquad \qquad \qquad \qquad \qquad 
+\frac{1}{\sigma_Y}\frac{\sum_{r=1}^{[ns]}\alpha_{(r:n)}(\mu_Y-\widehat{\mu}_Y)}{\sum_{r=1}^{[ns]}\alpha_{(r:n)}^2} \Bigg\}.
\label{eqn:s_stat}
\end{eqnarray}
Then, equation ($\ref{eqn:s_stat}$) can be written as
\begin{equation}\nonumber
\frac{1}{\sigma_Y}\frac{\sqrt{n}\sqrt{n\Psi_n^{\rm I}(1)}}{\sum_{r=1}^{[ns]}\alpha_{(r:n)}^2}{\rm U}(s)+
\rho\sqrt{n}{\rm V}(s)+\frac{\sqrt{n}}{\sigma_Y}{\rm R}(s),
\end{equation}
where 
$\Psi_n^{\rm I}(1)=\sum_{r=1}^n\alpha_{(r:n)}^2\sigma_{(r:n)}^2/n$ and
\begin{eqnarray}
{\rm U}(s) &=& \frac{1}{\sqrt{n \Psi_n^{\rm I}(1)}}\sum_{r=1}^{[ns]} \alpha_{(r:n)} \big( Y_{[r:n]} - m (X_{(r:n)}) \big)
\nonumber\\
&=&
\frac{1}{\sqrt{n \Psi_n^{\rm I}(1) } }\sum_{r=1}^{[ns]}
\left( \frac{  {\rm E}( X_{(r:n)} ) -\mu_X} { \sigma_X}   \right) \Big( Y_{[r:n]} - m (X_{(r:n)}) \Big), \label{eqn:psum}
\end{eqnarray}
with
$m(X_{(r:n)})={\rm E} \big(Y \big| X_{(r:n)} \big)=\mu_Y + \rho \sigma_Y (X_{(r:n)}-\mu_X)/\sigma_X$ and 
\begin{equation} \nonumber 
 {\rm V}(s)=\frac{\sum_{r=1}^{[ns]} \alpha_{(r:n)} \big( (X_{(r:n)}-\mu_X)/\sigma_X - \alpha_{(r:n)} \big)}{\sum_{r=1}^{[ns]} \alpha_{(r:n)}^2},~~{\rm R}(s)=\frac{\sum_{r=1}^{[ns]}\alpha_{(r:n)}(\mu_Y-\widehat{\mu}_Y)}{\sum_{r=1}^{[ns]}\alpha_{(r:n)}^2}.
\end{equation}
Since ${\rm R}(s)$ converges in probability to 0, we only consider the ${\rm U}(s)$ and ${\rm V}(s)$.
Thus, the proof of the theorem is based on the functional
central limit theorem for two partial sums of rank statistics, ${\rm U}(s)$ and ${\rm V}(s)$.

We first consider
the asymptotic distribution of the process of taking the weighted partial sum of the induced rank
statistic, which is
\begin{eqnarray}
{\rm U}(s) &=& \frac{1}{\sqrt{n \Psi_n^{\rm I}(1)}}\sum_{r=1}^{[ns]} \alpha_{(r:n)} \big( Y_{[r:n]} - m (X_{(r:n)}) \big)
\nonumber\\
&=&
\frac{1}{\sqrt{n \Psi_n^{\rm I}(1) } }\sum_{r=1}^{[ns]}
\left( \frac{  {\rm E}( X_{(r:n)} ) -\mu_X} { \sigma_X}   \right) \Big( Y_{[r:n]} - m (X_{(r:n)}) \Big). \label{eqn:psum}
\end{eqnarray}
The main finding of \citet{Bhattacharya:1974} is the
conditional independence of $Y_{[1:n]},\ldots,Y_{[n:n]}$ given
$X_1,X_2,\ldots,X_n$ (or equivalently, $X_{(1:n)}, X_{(2:n)}, \ldots, X_{(n:n)}$). \\Thus, given $\mathcal{A}=
\sigma\big( X_1,X_2,\ldots,X_n, \ldots \big)$, (\ref{eqn:psum}) can be read as
\begin{equation} \label{eqn:psum2}
S_{nk}=\frac{1}{\sqrt{n \Psi_n^{\rm I}(1)}}\sum_{r=1}^{k} \alpha_{(r:n)} \sigma_{(r:n)} u_r, \quad k=1,2,\ldots,n,
\end{equation}
 where the $u_r$ are independent, with mean $0$ and variance $\sigma_{(r:n)}^2$. By applying the basic concept of Skorokhod embedding 
\citep{Shorack:2009}, we
obtain a sequence of stopping times $\tau_{n1},\tau_{n2},\ldots,\tau_{nn}$ such that
\begin{itemize}
\item
these stopping times are conditionally
independent given $\mathcal{A}$,
\item ${\rm E} \big(\tau_{nk} \big| \mathcal{A} \big)= \sum_{r=1}^k \alpha_{(r:n)}^2 \sigma_{(r:n)}^2 \big/ \{ n \Psi_n^{\rm I}(1) \} $,
\item ${\rm var}  \big(\tau_{nk} \big| \mathcal{A} \big)= \sum_{r=1}^k \alpha_{(r:n)}^4
{\rm E} \big\{ \big( Y_{[r:n]} - m(X_{(r:n)}) \big)^4  \big| \mathcal{A} \big\} \big/ \{ n \Psi_n^{\rm I}(1) \} ^2  < \infty$, and
\item $\big(S_{n1}, S_{n2},\ldots,S_{nn} \big)$ has the same distribution as
$\big(W(\tau_{n1}), W(\tau_{n1}+\tau_{n2}),\ldots,W (\tau_{n1}+\tau_{n2}+\cdots+\tau_{nn}) \big)$, where $\big\{W(s), ~s \in [0, \infty) \big\}$ is
conventional Brownian motion.
\end{itemize}
We now consider the embedded partial-sum process $\big\{W_n(s): 0 \le s \le 1 \big\}$ that is defined by $W_n(s) =S_{n[ns]}$.
As in \citet{Bhattacharya:1974}, it suffices to show that
\begin{equation} \label{eqn:donsker1}
\sup_{0 \le s \le 1} \left|\frac{1}{n}\sum_{r=1}^{[ns]} \tau_{nr}  - \frac{\Psi_n^{\rm I}(s)}{\Psi_n^{\rm I}(1)} \right|
\end{equation}
converges to $0$ probability.

For each $s \in [0,1]$, the strong law of large numbers states that $\big(1 \big/ n \big)\sum_{r=1}^{[ns]} \tau_{nr} $
almost certainly converges to $\Psi_{\infty}^{\rm I}(s)\big/\Psi_{\infty}^{\rm I}(1)$.  Both
$\big(1 \big/n \big) \sum_{r=1}^{[ns]} \tau_{nr}$ and $\Psi_n^{\rm I}(s)\big/\Psi_n^{\rm I}(1)$ are increasing functions
of $s$. Thus, using the same arguments \citep[pp. 62]{Shorack:2009}, we find
that their sup difference also converges to $0$.

Second, 
\begin{eqnarray} 
\sqrt{n} {\rm V}(s) &=& \frac{1}{\Phi_n(s)}  \frac{1}{\sqrt{n}} \left\{\sum_{r=1}^{[ns]} \alpha_{(r:n)} \bigg(  \frac{X_{(r:n)}-\mu_X}{\sigma_X} - \alpha_{(r:n)} \bigg) \right\}
\end{eqnarray}
is a linear statistic of order statistics and converges to the normal distribution with mean $0$ 
and variance $\Psi_{\infty}^{\rm II}(s) \big/ \Phi_{\infty}^2(s)$ \citep[Theorem 11.4]{David:2003}. Here, we remark that both $\Psi_{\infty}^{\rm II}(s)$ and $\Phi_{\infty}(s)$ can also be written as functionals of the 
distribution of $X$, as shown in \citep{David:2003}.

Finally, summing the asymptotic results of ${\rm U}_n(s)$ and ${\rm V}_n(s)$, we find that 
$
\sqrt{n} \big( \widehat{\rho}(s) - \rho \big)$ converges to the normal distribution with mean $0$
and variance 
\begin{equation}\nonumber
\frac{\Psi_{\infty}^{\rm I}(s)/\sigma_Y^2+\rho^2\Psi_{\infty}^{\rm II}(s)}{\Phi_{\infty}^2(s)}
\end{equation}
This concludes the proof.

\subsection*{A.2 Proof of Theorem 1}

We first decompose the sample correlation between $\overline{S}_n(r)$ and $Y_{[r:n]}$ as ${\rm A}+{\rm B}+{\rm C}$:
\begin{eqnarray}
{\rm A} &=& \frac{1}{n}  \sum_{r=1}^n \overline{S}_n(r) \big\{ Y_{[r:n]} - m(X_{(r:n)})\big\}, \nonumber\\
{\rm B}&=& \frac{1}{n}   \sum_{r=1}^n \overline{S}_n(r) \big\{m(X_{(r:n)}) - \mu_Y \big\}, \nonumber\\
{\rm C} &=&  \frac{1}{n}  \sum_{r=1}^n \overline{S}_n(r) \big\{ \mu_Y - \overline{Y}_n \big\}, \nonumber
\end{eqnarray}
where $m(X_{(r:n)})=\mu_Y+\rho \sigma_Y X_{(r:n)}$. In below, we compute the limit of each ${\rm A}, {\rm B}$, 
and ${\rm C}$. First, similarly to the convergence of ${\rm U}(s)$ (with $s=1$) in Appendix A, we 
can show that $\sqrt{n} {\rm A}$ converges in distribution to a normal random variable and, thus, 
${\rm A}$ converges to $0$ in probability. Second, similarly to the convergence of ${\rm R}(s)$ (with $s=1$) 
in Appendix A, we can show that $\sqrt{n} {\rm C}$ converges in distribution to a normal random variable and, thus, 
${\rm C}$ converges to $0$ in probability. Lastly, 
\begin{eqnarray}
 {\rm B} &=& \frac{1}{n}   \sum_{r=1}^n \overline{S}_n(r) \big\{m(X_{(r:n)}) - \mu_Y \big\} \nonumber \\
 &=&  \rho \sigma_Y \cdot \frac{1}{n}   \sum_{r=1}^n \overline{S}_n(r) \frac{X_{(r:n)} - \mu_X}{\sigma_X} \nonumber \\
 &=&   \rho \sigma_Y \cdot \frac{1}{n}   \sum_{r=1}^n \overline{S}_n(r) \left\{ \frac{X_{(r:n)} - \mu_X}{\sigma_X}  - \alpha_{(r:n)} \right\} +   \rho \sigma_Y \cdot \frac{1}{n}   \sum_{r=1}^n \overline{S}_n(r) \alpha_{(r:n)}, \nonumber
\end{eqnarray}
 whose first term converges to $0$ in probability similarly to the convergence of ${\rm V}(s)$ (with $s=1$) in Appendix A. 
 Hence, ${\rm B}$ converges in probability to the limit of 
 \begin{equation}\label{eqn:inner}
  \rho \sigma_Y \cdot \frac{1}{n}   \sum_{r=1}^n \overline{S}_n(r) \alpha_{(r:n)}.
 \end{equation} 
 Since $\sum_{r=1}^n \overline{S}^2_n(r)=1$ and $\sum_{r=1}^n \alpha^2_{(r:n)}$ approaches $1$, (\ref{eqn:inner}) 
 is maximized when $\overline{S}_n(r)=\alpha_{(r:n)}$ in asymptotic.

\end{document}